\documentclass{article}

\usepackage{microtype}
\usepackage{graphicx}
\usepackage{subfigure}
\usepackage{algorithm} 
\usepackage{algorithmic}  
\usepackage[algo2e]{algorithm2e}  
\usepackage{booktabs} 
\usepackage[utf8]{inputenc}
\usepackage[english]{babel}
\usepackage{amsthm}
\usepackage{amsmath}
\usepackage{amsfonts}
\usepackage{bm}

\newtheorem{theorem}{Theorem}[section]

\newtheorem*{remark}{Remark}
\SetKwInput{KwInput}{Input}
\SetKwInput{KwOutput}{Output}
\usepackage{hyperref}

\usepackage[accepted]{icml2019}

\icmltitlerunning{Multi-Agent Deep Reinforcement Learning for Liquidation Strategy Analysis}

\begin{document}

\twocolumn[
\icmltitle{Multi-Agent Deep Reinforcement Learning for Liquidation Strategy Analysis}



\icmlsetsymbol{equal}{*}

\begin{icmlauthorlist}
\icmlauthor{Wenhang Bao}{stats}
\icmlauthor{Xiao-Yang Liu}{ee}
\end{icmlauthorlist}

\icmlaffiliation{stats}{Department of Statistics, Columbia University , New York, US}, 
\icmlaffiliation{ee}{Electrical Engineering, Columbia University, New York, US}

\icmlcorrespondingauthor{Wenhang Bao}{wb2304@columbia.edu}
\icmlcorrespondingauthor{Xiao-Yang Liu}{xl2427@columbia.edu}

\icmlkeywords{Multi-Agent, Reinforcement Learning, Liquidation, DDPG, Actor-Critic}

]



\printAffiliationsAndNotice{}  

\begin{abstract}

Liquidation is the process of selling a large number of shares of one stock sequentially within a given time frame, taking into consideration the costs arising from market impact and a trader's risk aversion. The main challenge in optimizing liquidation is to find an appropriate modeling system that can incorporate the complexities of the stock market and generate practical trading strategies. In this paper, we propose to use multi-agent deep reinforcement learning model, which better captures high-level complexities comparing to various machine learning methods, such that agents can learn how to make best selling decisions. First, we theoretically analyze the Almgren and Chriss model and extend its fundamental mechanism so it can be used as the multi-agent trading environment. Our work builds the foundation for future multi-agent environment trading analysis. Secondly, we analyze the cooperative and competitive behaviors between agents by adjusting the reward functions for each agent, which overcomes the limitation of single-agent reinforcement learning algorithms. Finally, we simulate trading and develop optimal trading strategy with practical constraints by using reinforcement learning method, which shows the capabilities of reinforcement learning methods in solving realistic liquidation problems.

\end{abstract}

\section{Introduction}
\label{sect:Introduction}

Liquidation, as one kind of stock trading, is one of the main functions of financial institutes, and the ability to minimize selling cost and manage risk level would be a key indicator of their financial performance. Therefore, effective trading strategy is of great importance. Financial institutes are updating their strategies recently, by making use of advanced research results or cutting-edge technologies. However, there are several challenges. First, liquidation of a large number of stock shares would have huge impact on the market, making the environment difficult to predict. Secondly, current methods for static environment ignore the dynamic and interactive nature of the stock market. Thirdly, the trading cost of liquidation depends on the stock market, and researchers are usually not able to collect enough historical events data to obtain practical trading insights.

\begin{figure}
    \centering
    \includegraphics[scale = 0.4]{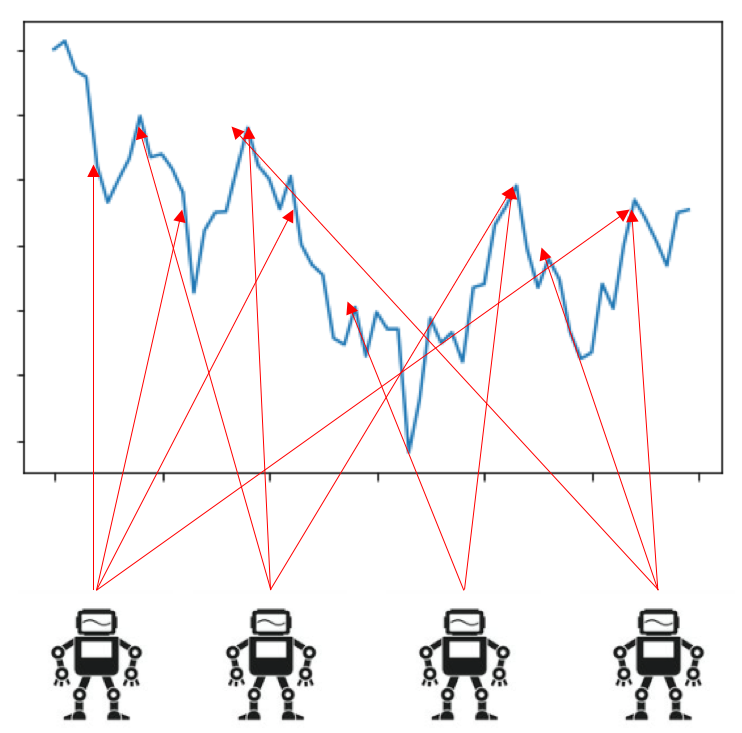}
    \caption{Liquidation: multiple agents sell stocks in the market, and their selling decisions would affect each others' selling cost}
    \label{fig:multi_agent}
\end{figure}

Financial modeling and machine learning are two popular approaches in developing trading strategies, but both of them have limitations. For the past years, financial institutes rely on experienced traders to minimize trading cost and manage liquidation risk. Also, researchers build mathematical and financial models to help develop liquidation strategies \cite{gomber2011high,Brogaard10highfrequency}. However, mathematical and financial modelling methods rely on their assumptions, which usually over-simplify the problem. Most recently, researchers started to adopt machine learning methods as well. 

Reinforcement learning (RL), one type of machine learning methods, consists of agents interacting with the environment to learn an optimal policy by trail and error for sequential decision-making problems \cite{sutton2018reinforcement,van2016deep}. While most of the successes of RL have been in the single agent domain, where modelling or predicting the behavior of other actors in the environment is not considered, the obtained trading strategy \cite{xiong2018practical} ignores the stochastic and interactive nature of the trading market. A more general scenario would be that multiple organizations or customers want to liquidate their assets under certain market conditions at the same time. Therefore, the trading market would have multiple players or institutes with similar objectives \cite{bansal2017emergent,tampuu2017multiagent}, and the behavior of one agent would affect other agents' behaviors \cite{yang2018mean}, as shown in Fig \ref{fig:multi_agent}. Another scenario would be even if there is only one company working on the liquidation of one stock, but still there could be multiple traders and each of them be responsible for a certain percentage of shares to sell. 

This calls for the demand of applying multi-agent reinforcement learning methods to the financial industry, which has not been well-studied as much as single-agent reinforcement learning. There are attractive successes of multi-agent deep reinforcement learning in the fields of gaming playing \cite{silver2016mastering,mnih2015human}, robotics and financial trading system \cite{yu2019model,buehler2019deep}. The main benefit of using reinforcement learning for liquidation is that mathematical models or hard-coded trading strategies can be avoided. Reinforcement learning agent would learn the trading strategy on its own. In addition, a simulated environment would allow agents to adapt to different market conditions and trade stocks, and obtain far more experience than human traders could obtain in real financial market \cite{schaul2015prioritized,foerster2017stabilising}. Last but not the least,  multi-agent reinforcement learning algorithms can take into account high-level environment complexities \cite{hendricks2014reinforcement} and derive more practical liquidation strategies accordingly. 

The main contribution of this paper is the analysis of the multi-agent trading environment, the impact analysis of coordinated relationship between agents, and the derivation of liquidation strategies. Ideally, if the multi-agent environment is complex enough to incorporate all potential players' behaviors, there would be no noise in the stock market, as all orders are generated by players, and all players' behaviors are modelled systematically by the multi-agent system. We build a simplified version of the multi-agent environment, which is the foundation of more complicated environments. First, we extend the model proposed by Almgren and Chriss \cite{almgren2001optimal} to the multi-agent environment and provide mathematical proofs. We make use of reinforcement learning to verify our theorems and conclude with the necessity to use multi-agent reinforcement learning instead of conventional single-agent reinforcement learning algorithms to analyze the liquidation problem. Secondly, we demonstrate how agents learn to cooperate or compete with each other by defining proper reward functions, analyze how these agents influence each other as well as the environment as a whole, which cannot be analyzed by a single-agent environment, but of great importance to financial institutes. Thirdly, we derive trading strategies for each agent in a simulated multi-agent environment. This demonstrates the capabilities of reinforcement learning algorithms in learning and developing practical liquidation strategies. 

The remainder of this paper is organized as follows. Section \ref{sect:Problem Statement} describes the liquidation problem and reviews the Almgren and Chriss model that is used to simulate the  market environment. Section \ref{sect:Solution} introduces the detailed settings of multi-agent reinforcement learning. Section \ref{sect:theorem} is about the extension of the multi-agent market environment. Section \ref{sect:Results} presents the experimental results where we demonstrate how agents would behave in cooperative or competitive relationships and how to derive liquidation strategy. Section \ref{sect:Conclusion} concludes this paper and points out some future direction. Code is available at: https://github.com/WenhangBao/Multi-Agent-RL-for-Liquidation

\section{Problem Description}
\label{sect:Problem Statement}
In this section, we first describe the liquidation problem and explain why it is feasible to use reinforcement learning algorithms to address it. Then we describe the Almgren and Chriss model or the trading environment. 

\subsection{Optimal Liquidation Problem}
\label{Liquidation}
We consider a liquidation trader who aims to sell $X$ shares of one stock within a time frame $T$. Liquidator's personal characteristics, such as risk aversion level $\lambda$, would remain unchanged throughout the process. The trader can either sell or not sell stocks, but cannot buy any stock during the time frame $T$. On the last day of the time frame, the liquidation process ends and the number of shares should be $0$. Since the trading volume is tremendous, the market price $P$ will drop during selling, temporarily or permanently, potentially resulting in enormous trading costs. 

The trader or the representative financial institute seeks to find an optimal selling strategy, minimizing the expected trading cost $E(X)$, or called \textit{implementation shortfall}, subject to certain optimization criterion. The trader would know all the environment information includes price, historical price and number of trading days remaining. If there are $J$ traders, they would not know other traders' information. For instance, they would not know other traders' remaining shares or risk aversion levels.  

Based on the assumption that the trading would have market impacts as well as that agents and environment are interactive, it is feasible to train agents in the environment and derive liquidation strategies with reinforcement learning algorithms \cite{yang2018practical}. 

\subsection{Environment Model for the Simulation}
\label{ACmodel}
The problem of an optimal liquidation strategy is investigated by using the Almgren-Chriss market impact model \cite{almgren2001optimal} on the background that the agents liquidate assets completely in a given time frame. The impact of the stock market is divided into three components: unaffected price process, permanent impact, and temporary impact. The stochastic component of the price process exists, but is eliminated from the mean-variance. The price process permits linear functions of permanent and temporary price. Therefore, the model serves as the trading environment such that when agents make selling decisions, the environment would return price information. 

The price process of the Almgren and Chriss model \cite{almgren2001optimal} is as follows:

\begin{itemize}
    \item Price under temporary and permanent impact
    $$
    P_k = P_{k-1} + \sigma \tau^{1/2} \xi_k -\tau g(\frac{n_k}{\tau}), k = 1,\ldots,N
    $$
where $\sigma$ represents the volatility of the stock, $\xi_k$ are random variables with zero mean and unit variance, $g(v)$ is a function of the average rate of the trading, $v = n_k/ \tau$ during time interval $t_{k-1}$ to $t_k$, $n_k$ is the number of shares to sell during time interval $t_{k-1}$ to $t_k$, $N$ is the total number of trades and $\tau = T/N$.
    \item Inventory process: $x_{t_k} = X - \sum_{j=1}^{k}n_j$, where $x_{t_k}$ is the number of shares remaining at time $t_k$, with $x_T = 0$.
    \item Linear permanent impact function
    $g(v) = \gamma v$, where $v = \frac{n_k}{\tau}$.
    \item Temporary impact function
    $h(\frac{n_k}{\tau}) = \epsilon~\text{sgn}(n_k) + \frac{\eta}{\tau} n_k$, where a reasonable estimate of $\epsilon$ is the fixed costs of selling, and $\eta$ depends on internal and transient aspects of the market micro-structure.
    \item Parameters $\sigma, \gamma, \eta, \epsilon$, time frame $T$, number of trades $N$ are set at $t = 0$.
\end{itemize}

\section{Deep Reinforcement Learning Approach}
\label{sect:Solution}

We model the liquidation process as a Markov decision process (MDP), and then formulate the multi-agent setting we used to resolve the problem. The training diagram is also covered, which explains how multiple agents interact and learn from environment in details. We use implementation shortfall as the metric of selling cost, and the properties of MDP process allows us to define the goal as minimizing the expected implementation shortfall.  

\subsection{Liquidation as a MDP Problem}
\label{DPL:liquidatoin}
 Consider the stochastic and interactive nature of the trading market, we model the stock trading process as a Markov decision process, which is specified as follows:
\begin{itemize}
    \item State $s = [\bm{r},m,l]$: a set that includes the information of the log-return $\bm{r} \in \mathbb{R}_+^D$, where $D$ is the number of days of log-return, and the remaining number of trades $m$ normalized by the total number of trades, the remaining number of shares $l$, normalized by the total number of shares. The log-returns capture information about stock prices before time $t_k$, where $k$ is the current step. It is important to note that in real world trading scenarios, this state vector may hold more variables. 
    \item Action $a$: we interpret the action $a_k$ as a selling fraction. In this case, the actions will take continuous values in between 0 and 1.
    \item Reward $R(s,a)$: to define the reward function, we use the difference between two consecutive utility functions. The utility function is given by:
    \begin{align}
        U(\bm{x}) &= E(\bm{x}) + \lambda V(\bm{x}) ,\label{eq:1}\\
        E(\bm{x}) &= \sum_{k=1}^{N}\tau x_k g(\frac{n_k}{\tau}) + \sum_{k=1}^{N}n_k h(\frac{n_k}{\tau}),\label{eq:2}\\
        V(\bm{x}) &= \sigma^2 \sum_{k=1}^{N}\tau x_k^2,\label{eq:3}
    \end{align}
        
where $\lambda$ is the risk aversion level, and $\bm{x}$ is the trading trajectory or the vector of shares remaining at each time step $k,~ 0 \le t_k \le T$. After each time step, we compute the utility using the equations for $E(\bm{x})$ and $V(\bm{x})$ from the Almgren and Chriss model for the remaining time and inventory while holding parameter $\lambda$ constant. Denotes the optimal trading trajectory computed at time $t$ by $\bm{x}_t^*$, we define the reward as:
        \begin{equation}
            R_{t} = {U_t(\bm{x}_t^*)-U_{t+1}(\bm{x}_{t+1}^*)} .
        \end{equation}
    \item Policy $\pi (s)$: The liquidation strategy of stocks at state $s$. It is essentially the distribution of selling percentage $a$ at state $s$.
    \item Action-value function $Q_\pi(s,a)$: the expected reward achieved by action $a$ at state $s$, following policy $\pi$.
\end{itemize}

\subsection{Multi-agent Reinforcement Learning Setting}
\label{DPL:Setting}
The advantages of multi-agent over single-agent reinforcement learning is the ability to incorporate high-level complexities in the system. The single-agent environment is a special case where the number of agents $J=1$. It simplifies the problem and would automatically inherit all properties from the multi-agent environment. Following the MDP configuration in the last section, we specify our multi-agent reinforcement learning setting as follows:

\begin{itemize}
    \item States $s = [\bm{r},m,\bm{l}]$ : in a multi-agent environment, the state vector should have information about the remaining stocks of each agent. Therefore, in a $J$ agents environment, the state vector at time $t_k$ would be:
        $$
        [r_{k-D},\ldots,r_{k-1},r_k,m_k,l_{1,k},\ldots,l_{J,k}],
        $$
where
    \begin{itemize}
        \item $r_k = \log(\frac{P_k}{P_{k-1}})$ is the log-return at time $t_k$.
        \item $m_k = \frac{N_k}{N}$ is the number of trades remaining at time $t_k$ normalized by the total number of trades.
        \item $l_{j,k} = \frac{x_{j,k}}{X_j}$ is the remaining number of shares for agent $j$ at time $t_k$ normalized by the total number of shares. 
    \end{itemize}

    \item Action $a$: using the interpretation in Section \ref{DPL:liquidatoin}, we can determine the number of shares to sell for each at each time step using:
        $$
        n_{j,k} = a_{j,k} \times x_{j,k} ,
        $$
where  $x_{j,k}$ is the number of remaining shares at time $t_k$ for agent $j$.

    \item Reward $R(s,a)$: denotes the optimal trading trajectory computed at time $t$ for agent $j$ by $x_{j,t}^*$, we define the reward as:
        \begin{equation}
            R_{j,t} = {U_{j,t}(\bm{x}_{j,t}^*)-U_{j,t+1}(\bm{x}_{j,t+1}^*)}.
        \end{equation}
        
    \item Observation $O$: Each agent only observes limited state information \cite{omidshafiei2017deep}. In other words, in addition to the environment information, each agent only knows its own remaining shares, but not other agents' remaining shares. The observation vector at time $t_k$ for agent $j$ is:
        $$
        O_{j,k} = [r_{k-D},\ldots,r_{k-1},r_k,m_k,l_{j,k}].
        $$
\end{itemize}

\subsection{Deep Reinforcement Learning Algorithm}
We adopt the Actor-Critic \cite{mnih2016asynchronous,lowe2017multi} method that uses neural networks to approximate both the Q-value and the action. The critic learns the Q-value function and uses it to update actor's policy parameters. The critic network estimates the expected return of a state-action pair.  
The actor brings the advantage of computing continuous actions without the need of a Q-value function, while the critic supplies the actor with knowledge of the performance. The actor network has state $s$ as input and returns action $a$ directly. Actor-critic methods usually have good convergence properties, in contrast to critic-only methods. 

\begin{algorithm}[H]
    \caption{DDPG-Based Multi-agent Training}
    \KwInput{number of episodes $M$, time frame $T$, minibatch size $N$, learning rate $\lambda$, and number of agents $J$}
    \begin{algorithmic}[1]
    \label{alg:DDPG}
        \FOR{$j = 1, J$ \% initialize each agent separately}
            \STATE Randomly initialize critic network $Q_j(O_j,a|\theta_j^Q)$ and actor network $\mu_j(O_j|\theta_j^\mu)$ with random weight $\theta_j^Q$ and $\theta_j^\mu$ for agent $j$;
            \STATE Initialize target network $Q'_j$ and $\mu'_j$ with weights $\theta_j^{Q'} \leftarrow \theta_j^{Q}$, $\theta_j^{\mu'} \leftarrow \theta_j^{\mu}$ for each agent $j$;
            \STATE Initialize replay buffer $B_j$ for each agent $j$;
        \ENDFOR
        \FOR {episode $= 1, M$}
            \STATE Initialize a random process $\mathcal{N}$ for action exploration;
            \STATE Receive initial observation state $s_0$;
            \FOR{$t = 1, T$}
                \FOR{$j = 1, J$ \%train each agent separately}
                    \STATE {Select action $a_{j,t} = \mu_j(O_{j,t}|\theta_j^\mu) + \mathcal{N}_t$ according to the current policy and exploration noise;}
                \ENDFOR
                \STATE {Each agent executes action $a_{j,t}$;}
                \STATE {Market state changes to $s_{t+1}$};
                \STATE {Each agent observes reward $r_{j,t}$ and observation $O_{j,t+1}$};
                \FOR{$j = 1, J$}
                    \STATE {Store transition ($O_{j,t}$, $a_{j,t}$, $r_{j,t}$, $O_{j,t+1})$ in $B_j$;}
                    \STATE {Sample a random minibatch of $N$ transitions ($O_{j,i}$ , $a_{j,i}$ , $r_{j,i}$ , $O_{j,i+1}$) from $B_j$;}
                    \STATE {Set \\
                    $y_{j,i} = r_{j,i}+\gamma Q'_j (s_{t+1}, \mu'_j (O_{j,i+1}|\theta_j^{\mu'}|\theta_j^{Q'}))$ \\
                    for $i = 1, \ldots, N$;}
                    \STATE {Update the critic by minimizing the loss: $L = \frac{1}{N}\sum_i(y_{j,i} -Q_j(O_{j,i},a_{j,i}|\theta_j^Q))^2$;}
                    \STATE {Update the actor policy by using the sampled policy gradient:
                        \begin{multline*}
                            \nabla_{\theta^\mu} \pi \approx \frac{1}{N}\sum_i \nabla_a Q_j(O,a|\theta_j^Q)|_{O = O_{j,i},a = \mu_j(O_{j,i})}\\
                            \times \nabla_{\theta^\mu} \mu_j(O_j|\theta^\mu)|_{s_i};
                        \end{multline*}
                    }
                    \STATE {Update the target networks:\\
                    $\theta_j^{Q'} \leftarrow \tau \theta_j^Q + (1-\tau)\theta_j^{Q'},$\\
                    $\theta_j^{\mu'} \leftarrow \tau \theta_j^\mu + (1-\tau)\theta_j^{\mu'}.$}
                \ENDFOR
            \ENDFOR
        \ENDFOR
    \end{algorithmic}
\end{algorithm}

The Deep Deterministic Policy Gradients (DDPG) algorithm \cite{lillicrap2015continuous} is one example of an actor-critic method. We will use DDPG to generate the optimal execution strategy of liquidation. DDPG uses three skills to make sure it gets converged experimental results: experience replay buffer, learning rate and exploration noise. Experienced replay method \cite{wang2016sample} enables the stochastic gradient decent method and removes correlations between consecutive transitions. Learning rate controls the updating speed of the neural network. Exploration noise addresses the exploration and exploitation trade-off. With these training skills, the agent would learn from trail and error and find the optimal trading trajectory that minimizes the trading cost. In other words, we will use the DDPG algorithm or Alg. \ref{alg:DDPG} to solve the optimal liquidation problem. 

\section{Performance Analysis}
\label{sect:theorem}
Here we extend the classical environment model to multi-agent scenario for liquidation problem analysis.

\subsection{Optimal Multi-agent Liquidation Shortfall}
\label{results:Expected_implementation_Shortfall}

\begin{theorem}
\label{First_theorem_of_multi_agents_trading}
In a multi-agent environment with $J$ agents where each agent has $X_j$ shares to sell within a given time frame $T$, the total expected shortfall is larger than or equal to the sum of expected shortfall that these agents would obtain if they are in single-agent environment, such that:
        \begin{equation}
            \sum_{j=1}^{J}E(X_j) \le E(\sum_{j=1}^{J}X_j),
        \end{equation}
where $E(X)$ is the expected implementation shortfall of liquidating $X$ shares of a stock.
\end{theorem}

\begin{proof}
According to the Almgren and Chriss model, (namely, equation (20) in \cite{almgren2001optimal}), the optimal expected shortfall is:
    \begin{equation}\label{eq:shortfall}
        E(X) = \frac{1}{2} \gamma X^2 + \epsilon X + \tilde \eta \phi X^2,
    \end{equation}
where $X$ is the initial stock size and $\phi$ is a parameter related with environment setting but unrelated with the stock size $X$.

Therefore,
    \begin{align*}
    E(\sum_{j=1}^{J}X_j) &= \frac{1}{2} \gamma (\sum_{j=1}^{J}X_j)^2 + \epsilon \sum_{j=1}^{J}X_j + \tilde \eta (\sum_{j=1}^{J}X_j)^2 \phi\\
    &\ge \frac{1}{2} \gamma \sum_{j=1}^{J}X_j^2 + \epsilon \sum_{j=1}^{J}X_j +\tilde \eta \sum_{j=1}^{J}X_j^2 \phi\\
    &=\sum_{j=1}^{J}E(X_j).
    \end{align*}

\end{proof}

\subsection{Multi-agent Interaction}

\begin{theorem}
\label{Lemma_Second_theorem}
In a two-agent environment where agent $1$ has risk aversion level $\lambda_1$ and agent $2$ has risk aversion level $\lambda_2$, where $\lambda_1 \neq \lambda_2$, and each of them has the same number of stocks to liquidate, the biased trajectories $x(\lambda_1)$ and $x(\lambda_2)$ would satisfy that
        $$
        \bm{x}^*(\lambda_1) \neq \bm{x}(\lambda_1),~\bm{x}^*(\lambda_2) \neq \bm{x}(\lambda_2),
        $$
where $\bm{x}^*(\lambda_1)$ and $\bm{x}^*(\lambda_2)$ are the optimal trading trajectories when they are the only player in the market.
\end{theorem}

\begin{remark}
\label{Second_theorem_of_multi_agents_trading}
In a multi-agent environment where each agent has risk aversion level $\lambda_j$, the actual trading trajectory $\bm{x}(\lambda_j)$ would be biased against the optimal trading trajectory. 
\end{remark}

\begin{proof}
According to (4) of the Almgren and Chriss model \cite{almgren2001optimal},
    $$
    V(\bm{x}) = \sigma^2 \sum_{k=1}^{N}\tau x_k^2
    $$
is irrelevant to either temporary or permanent price changes, where $x_k$ is the remaining shares remaining at time $t_k$.
The optimal trading trajectory is of the form:
    $$
    x_k = \frac{\sinh(\kappa(\lambda)(T-t_j)}{\sinh(\kappa(\lambda) T)} X,
    $$
where $\kappa(\lambda) = \frac{\lambda \sigma^2}{\eta(1-\frac{\gamma \tau}{2\eta})}$.

Let $X$ be the total stock size and agent $1$ and $2$ each has $\frac{1}{2} X$ shares, the utility function
    $$
    U(\bm{x}) = E(\bm{x}) + \lambda^*V(\bm{x})
    $$
is a quadratic function of the parameters $x_1,\ldots,x_{N-1}$, where $\lambda^*$ is the synthesized risk aversion level, $\bm{x}$ is the trading trajectory, and it could also be written as:
    $$
    U(\bm{x}) = E(\bm{x}) + \lambda_1 V(\bm{x}_1) + \lambda_2 V(\bm{x}_2),
    $$
where $\bm{x}_1,\bm{x}_2$ is the trading trajectory for agent $1,2$, respectively. Then:
    $$
    \frac{\partial U}{\partial x_k} = 2\tau \left \{(\frac{\lambda_1 + \lambda_2}{2}) \sigma^2 x_k -\tilde \eta \frac{x_{k-1}-2x_k+x_{k+1}}{\tau^2} \right \},
    $$
and $\frac{\partial U}{\partial x_k} = 0$ is equivalent to 
    
    \begin{equation}\label{eq:sinh}
        \frac{1}{\tau^2}(x_{k-1}-2x_k+x_{k+1}) = (\tilde \kappa^*)^2 x_k
    \end{equation}
with 
    $$
    \tilde \kappa^*= \frac{\frac{(\lambda_1 + \lambda_2)}{2}\sigma^2}{\eta(1-\frac{\gamma \tau}{2\eta})}, 
    $$
where the tilde denotes an $\mathbb O(\tau)$ correction; as $\tau \rightarrow 0 $, we have $\tilde \kappa \rightarrow \kappa$. 
Then, we know that the solution to (\ref{eq:sinh}) is:
    \begin{align*}
        x_k&= \frac{\sinh(\kappa^*(T-t_k))}{\sinh(\kappa T)} X\\
        &\neq \frac{\sinh(\kappa(\lambda_1)(T-t_k))}{\sinh(\kappa(\lambda_1) T)} \frac{1}{2}X+\frac{\sinh(\kappa(\lambda_2)(T-t_k))}{\sinh(\kappa(\lambda_2) T)} \frac{1}{2}X,
    \end{align*}

where the right-hand side is the total number of remaining shares at time $t$ if both agents follow their original trading trajectories, and that is not equal to the total remaining shares at time $t$ under optimal trading trajectory, which is the left-hand side of the function. In other words, their new trading trajectories would be biased.
\end{proof}

\section{Performance Evaluation}
\label{sect:Results}
We first describe the simulation environment in details, and then verify the Theorem \ref{First_theorem_of_multi_agents_trading} and Theorem \ref{Lemma_Second_theorem} by experiments. We then use reinforcement learning methods to demonstrate how agents learn to cooperate or compete with each other by defining proper reward functions, and analyze how the relationship would influence each individual player as well as the environment. Finally, we derive practical trading strategies in a multi-agent environment.

We implement a typical reinforcement learning workflow to train the actor and critic. We change the single-agent Almgren and Chriss model \cite{almgren2001optimal} settings to build the multi-agent environment. We adjust reward functions to manipulate agents' relationships. We use Alg. \ref{alg:DDPG} to find a policy that can generate the optimal trading trajectory with minimum implementation shortfall. We feed the states observed from our simulator to each agent. These agents first predict actions using the actor model and perform these actions in the environment. Then, environment returns their rewards and new states. This process continues for a given number of episodes. 

\subsection{Simulation Environment}
\label{Simulation_env}
This environment simulates stock prices that follow a discrete arithmetic random walk, and that the permanent and temporary market impact functions are linear functions of the rate of trading, as in the Almgren and Chriss model \cite{almgren2001optimal}. 

We set the total number of shares to $1$ million and the initial stock price to be $P_0 = 50$, which gives an initial portfolio value of  $\$50$ million dollars. The stock price has $12\%$ annual volatility, a bid-ask spread of ${1}/{8}$, the difference between ask price and bid price, and an average daily trading volume of 5 million shares. Assuming that there are 250 trading days in a year, this gives a daily volatility in stock price of  ${0.12}/{\sqrt{250}}\approx 0.8\%$ . We use a liquidation time frame of $T = 60 $ days and we set the number of trades $N = 60$. This leads to $\tau = \frac{T}{N} = 1$ , which means that we will be making one trade per day. These settings are changeable and can be adjusted to same day liquidation as well. 

For the temporary cost function, we set the fixed cost of selling to be ${1}/{2}$ of the bid-ask spread, so $\epsilon = {1}/{16}$. We set $\eta$ such that for each one percent of the daily volume we trade, the price impact equals to the bid-ask spread. For example, trading at a rate of  $5\%$  of the daily trading volume incurs a one-time cost on each trade of ${5}/{8}$. Under this assumption we have  $\eta = {{(1/8)}/ {(0.01 \times 5\times 10^6)}}=2.5\times 10^6$.

For the permanent costs, a common rule of thumb is that price effects become significant when we sell $10\%$ of the daily volume. Here, by "significant" we mean that the price depression is one bid-ask spread, and that the effect is linear for both smaller and larger trading rates, then we have $\gamma = {{(1/8)}/{(0.1\times 5 \times 10^6)}} = 2.5 \times 10 ^7$.

In all our experiments, we run the program for $10000$ episodes, unless it is specified. Also, we use the following reward definition: 
        \begin{equation}
            \tilde R_{j,t} = {{U_{j,t}(\bm{x}_{j,t}^*)-U_{j,t+1}(\bm{x}_{j,t+1}^*)}\over {U_{j,t}(\bm{x}_{j,t}^*)}},
        \end{equation}
which normalizes the reward.

\subsection{Theorem Verification}
\subsubsection{Optimal Liquidation Shortfall}

We first train one agent $A$ who would need to liquidate $1$ million shares of a stock. Then we train two agents $B_1$ and $B_2$ who have the same targets and $X_1,X_2 = 0.3,0.7$ million shares, respectively. Agents $A$, $B_1$ and $B_2$ have the same risk aversion level $\lambda_A = \lambda_{B_1} = \lambda_{B_2} = 1e-6$. As we can see from Fig. \ref{fig:Expected_Implementation_Shortfall}, the expected implementation shortfall $E(A)$ is larger than the sum of $E(B1)$ and $E(B2)$.

This result justifies Theorem \ref{First_theorem_of_multi_agents_trading}. The intuition behind Theorem \ref{First_theorem_of_multi_agents_trading} and Equation \ref{eq:shortfall} is that the total expected shortfall increases faster than the total number of stock shares. 

\begin{figure}
    \centering
    \includegraphics[scale = 0.4]{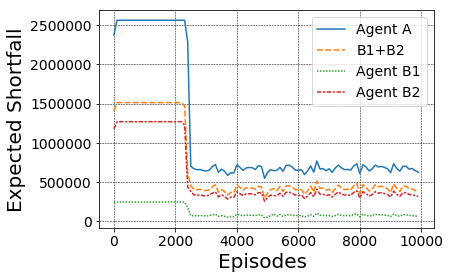}
    \caption{Comparison of expected implementation shortfalls: there are three agents $A, B1$ and $B2$. The expected shortfall of agent A is higher than the sum of two expected shortfalls $B_1$ and $B_2$.}
    \label{fig:Expected_Implementation_Shortfall}
\end{figure}
\subsubsection{Multi-agent Interaction}
Here we would like to analyze the trading trajectory of two agents, or Theorem \ref{Lemma_Second_theorem} as an illustration. We first train agent $A_1$ with risk aversion level $\lambda_{A_1} = 1e-4$ and agent $A_2$ with risk aversion level $\lambda_{A_2} = 1e-9$. Both $A_1$ and $A_2$ are trained separately in a single-agent environment. Then we train agent $B_1$ and $B_2$ with risk aversion level $\lambda_{B_1} = 1e-4, \lambda_{B_2} = 1e-9$, respectively, in a two-agent environment. All these agents have the same goal as defined in Section \ref{sect:Solution}. The trading trajectories of $A_1,A_2,B_1,B_2$ are shown in Fig. \ref{fig:Trading_trajectory}. 

Comparing to the single-agent environment, we can see that the trading trajectory of $B_1$ and $B_2$ are biased. Unlike the single-agent scenario, where they can sell their shares independently, now they have to take into consideration of other players in the market. The selling patterns of other agents would affect their liquidation strategy. The results not only justify Remark \ref{Second_theorem_of_multi_agents_trading} and Theorem \ref{Lemma_Second_theorem}, it also demonstrates the necessity of using multi-agents reinforcement learning algorithm to derive trading strategy. All traders are influencing each other when they are executing their own strategy. Therefore, training one agent in a single-agent environment over-simplifies the stochastic and dynamic nature of the stock market, as we have explained in Section \ref{sect:Introduction}. 

\begin{figure}
    \centering
    \includegraphics[scale = 0.4]{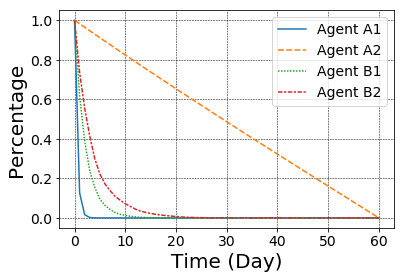}
    \caption{Trading trajectory: comparing to their original trading trajectories, their current trading trajectories are closer to each other when they are trained in a multi-agent environment.}
    \label{fig:Trading_trajectory}
\end{figure}

\subsection{Multi-agent Coordinated Relationship}
To analyze the emergence of a variety of coordinated behaviors, we adjust the rewarding schemes to change the relationship between agents. There are only two agents in this environment for illustration purpose. Each agent would be responsible for selling $0.5$ million shares of a stock. They share the same risk aversion level $\lambda = 1e-6$. The only difference between the next two experiments is the definition of the reward functions. Then we compare the sum of expected shortfalls with the expected shortfall trained independently, to evaluate how the relationship would affect the total as well as individual implementation shortfalls.

\subsubsection{Multi-agents Cooperation}
\label{results:Cooperation}
In this setting we want to analyze how agents would behave when they are in a cooperative relationship. Therefore, we adjust the reward function as follows:
        \begin{equation}
            \tilde R_{1,t}^*=\tilde R_{2,t}^* = {{\tilde R_{1,t}+\tilde R_{2,t}} \over {2}},
        \end{equation}
where $\tilde R_{j,t}^*$ are the new reward functions.

Both agents would be rewarded by the sum of their individual rewards. So the two agents would be fully cooperative to minimize implementation shortfall. The result is shown in Fig. \ref{fig:Cooperative_and_competative}. First, we notice that the sum of expected shortfall does not change much comparing to training two agents with reward function $\tilde R_{j,t}$. Secondly, new individual implementation shortfall $E^*(\bm{x}_j^*)$ does not change much comparing to the original implementation shortfall $E(\bm{x}_j^*)$, where $\bm{x}_j^*$ is the optimal trading trajectory.

\begin{figure}
    \centering
    \includegraphics[scale = 0.4]{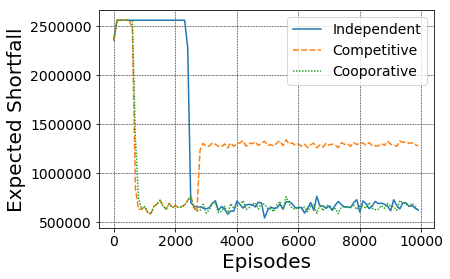}
    \caption{Cooperative and competitive relationships: if two agents are in cooperative relationship, the total expected shortfall is not better than training with independent reward functions. If two agents are in a competitive relationship, they would first learn to minimize expected shortfall, and then malignant competition leads to significant implementation shortfall increment.}
    \label{fig:Cooperative_and_competative}
\end{figure}

\subsubsection{Multi-agents Competition}
\label{results:Competition}
In this setting we want to analyze how agents would behave when they are in a competitive relationship. Therefore, we adjust the reward function as follows:

\begin{algorithmic}
        \IF{$\tilde R_{1,t} > \tilde R_{2,t}$}
        \STATE $\tilde R_{1,t}^* = \tilde R_{1,t}$,
        \STATE $\tilde R_{2,t}^* = \tilde R_{2,t} - \tilde R_{1,t}$,
        \ELSE 
        \STATE $\tilde R_{2,t}^* = \tilde R_{2,t}$,
        \STATE $\tilde R_{1,t}^* = \tilde R_{1,t} - \tilde R_{2,t}$,
        \ENDIF

\end{algorithmic}
where $\tilde R_{j,t}^*$ are the new reward functions. 

In this case, the agent that gets higher reward would keep it, but the agent with lower reward would be penalized. It would receive reward value equal to its original reward minus the higher reward, which is a negative value. As we can see from Fig. \ref{fig:Cooperative_and_competative} that the sum of expected shortfalls ends up with about twice as they are independent or in cooperative relationship. By looking at the snapshot of trading trajectory at 1500 episode and 10000 episode, we notice that at the 1500 episode, these two agents learn to maximize utility function defined in Section \ref{DPL:liquidatoin}. Both agents perform well and roughly have the same expected shortfall. However, as they are in a competitive relationship, after another 500 episodes of training, one agent learns to ourperform the other, which leads to the significant increment of the sum of expected shortfall, or $\sum_{j=1}^{2}E^*(\bm{x}_j^*) >\sum_{j=1}^{2}E(\bm{x}_j^*)$. The trading trajectory of the last episode shows that one agent learns to sell all its shares on $Day~1$. In addition, the expected shortfall for both agents increase, or $E^*(\bm{x}_j^*) > E(\bm{x}_j^*)$. 

We conclude that not only their overall performance is diminished, their individual performance is worse as well. None of them is winning from their mutual competition. 

\subsection{Liquidation Strategy Development}
We use reinforcement learning algorithm to develop trading strategies, given the trading trajectories of the competitors. 


\begin{figure}
    \centering
    \includegraphics[scale = 0.4]{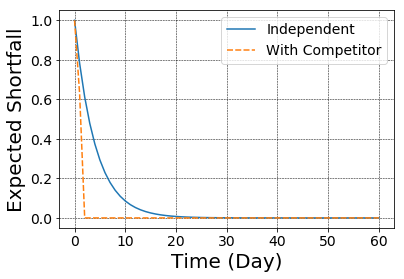}
    \caption{Trading trajectory: comparing to independent training, introducing a competitor makes the host agent learn to adapt to new environment and sell all shares of stock in the first two days.}
    \label{fig:Have_competitor}
\end{figure}

Here we introduce an agent who has $0.5$ million shares of stocks to sell and has risk aversion level $\lambda = 1e-9$. We have already seen in Fig. \ref{fig:Trading_trajectory} that the optimal trading trajectory for such agent would be a straight line, which means it sells a fixed amount of stocks everyday.

We train an agent who has another $0.5$ million shares of stocks to sell with risk aversion level $\lambda = 1e-6$. For comparison purpose, we also draw the optimal trading trajectory when the agent is trained independently in a single-agent environment. As we can see in Fig. \ref{fig:Have_competitor}, if there is no competitor, the optimal trajectory shows that the agent would complete the liquidation process in about $20$ days. After we introduced the competitor, the trading trajectory completely changed. Now the agent sells all its shares within the first $2$ days. The agent learns to avoid taking unnecessary risk by selling all shares in a quite short time, and let the competitor agent to bear the execution cost of price drop.   


\section{Conclusion and Future Work}
\label{sect:Conclusion}

\subsection{Contributions}
We have shown the single-agent environment over-simplifies the dynamic as well as the interactive nature of the stock market. All orders are generated by individual traders, and these traders act as game players, especially for a systematic trading problem like liquidation. 

In the present work, we extended the scope of Almgren and Chriss model \cite{almgren2001optimal} and used reinforcement learning method to verify it, which setups the foundation of the multi-agent trading environment. We illustrate the demand to use multi-agent environments to develop trading strategies. We analyze how fully cooperative and competitive relationship would affect the total and individual implementation shortfalls, respectively. We conclude that cooperative relationship is not better than independent one, and competitive relationship would hurt the overall and individual performance. Finally, we demonstrate the capability of reinforcement learning agent and derived optimal liquidation strategy for the host agent against its competitor. 

\subsection{Limitations}
As the goal of this paper is to analyze the environment and agents interaction, we keep simple setting as long as it is reasonable. Therefore, we did not build more complex neural network architectures, and our best expected shortfall after 10000 episodes of training is roughly $20\%$ higher than the optimal expected shortfall derived by the Almgren and Chriss model \cite{almgren2001optimal}. We can add more dynamic factors in the state vector. Also, advanced background models other than Amlgren and Chriss model could also be considered. While all these methods could potentially improve this work, we believe that at this moment they are not necessary for describing the nature of multi-agent trading environments and analyzing agents' behaviors, for this preliminary analysis of liquidation problem.

\subsection{Future Work}
Development of more realistic trading environment, including more dynamic factors such as news, general strategy and legal complaints, would make great contributions to financial analysis. A potential extension is the study of stock liquidation by considering the optimistic bull or pessimistic bear \cite{Xinyi2019ICML} or the anomaly events \cite{Xinyi2019KDD}. A potential application would be using predicted agents' behaviors to predict stock price movements, say LSTM \cite{Xinyi2019KDD}. 



\bibliographystyle{icml2019}

\end{document}